\documentclass[12pt,oneside]{article}
\usepackage[cp1250]{inputenc}
\usepackage{slashbox}
\usepackage{multirow}
\usepackage{graphicx}
\usepackage[usenames]{color} 
\usepackage{enumerate}
\usepackage[usenames]{color} 
\usepackage{amsmath, amsthm, amsfonts, amssymb}

\newtheorem{theorem}{Theorem}[section]

\newtheorem{corollary}[theorem]{Corollary}

\newtheorem{proposition}[theorem]{Proposition}
\newtheorem{conjecture}{Conjecture}

\begin{document}

\markboth{ }{}
\title{\bf Equitable Colorings of Corona Multiproducts of Graphs\footnote{Project has been 
partially supported by Narodowe Centrum Nauki under contract 
DEC-2011/02/A/ST6/00201}}
\date{}
\author{Hanna Furmańczyk\footnote{Institute of Informatics,\ University of Gdańsk,\ Wita Stwosza 57, \ 80-952 Gdańsk, \ Poland. \ e-mail: hanna@inf.ug.edu.pl},  \ Marek Kubale \footnote{Department of Algorithms and System Modelling,\ Technical University of Gdańsk,\ Narutowicza 11/12, \ 80-233 Gdańsk, \ Poland. \ e-mail: kubale@eti.pg.gda.pl} \ Vahan V. Mkrtchyan \footnote{Department of Informatics and Applied Mathematics,\ Yerevan State University, Armenia. e-mail: vahanmkrtchyan2002@ipia.sci.am}}

\markboth{H. Furmańczyk, M. Kubale, V. Mkrtchyan}{Equitable Colorings of Corona Multiproducts of Graphs}

\maketitle

\begin{abstract}
A graph is equitably $k$-colorable if its vertices can be partitioned into $k$ independent 
sets in such a way that the number of vertices in any two sets differ by at most one. 
The smallest $k$ for which such a coloring exists is known as the \emph{equitable chromatic 
number} of $G$ and denoted $\chi _{=}(G)$.
It is known that this problem is NP-hard in general case and remains so for corona graphs. In \cite{prod} Lin i Chang studied equitable coloring of Cartesian products of graphs. In this paper we consider the same model of coloring in the case of corona products of graphs. In particular, we obtain some results regarding the equitable chromatic number for $l$-corona product $G \circ ^l H$, where $G$ is an equitably 3- or 4-colorable graph and $H$ is an $r$-partite graph, a path, a cycle or a complete graph.  Our proofs are constructive in that they lead to polynomial algorithms for equitable coloring of such graph products provided that there is given an equitable coloring of $G$. Moreover, we confirm Equitable Coloring Conjecture for corona products of such graphs. This paper extends our results from \cite{hf}.
\end{abstract}


\section{Introduction}
All graphs considered in this paper are finite, connected and simple, i.e. undirected, loopless and without multiple edges.

If the set of vertices of a graph $G$ can be partitioned into $k$ (possibly empty) classes $V_1,V_2,....,V_k$ such that each $V_i$ is an independent set and the condition $||V_i|-|V_j||\leq 1$ holds for every pair ($i, j$), then $G$ is said to be {\it equitably k-colorable}. In the case, where each color is used the same number of times, i.e. $|V_i|=|V_j|$ for every pair ($i, j$), graph $G$ is said to be \emph{strongly equitably $k$-colorable}. The smallest integer $k$ for which $G$ is equitably $k$-colorable is known as the {\it equitable chromatic number} of $G$ and denoted $\chi_{=}(G)$ \cite{meyer}. 
Since equitable coloring is a proper coloring with additional condition, the inequality $\chi(G) \leq \chi_=(G)$ holds for any graph $G$.

In some discrete industrial systems we can encounter the problem of equitable 
partitioning of a system with binary conflict relations into conflict-free 
subsystems. Such situations can be modeled by means of equitable graph coloring. 
For example, in the garbage collection problem the vertices of the graph represent garbage collection routes and a pair of vertices is joined by an edge if the corresponding routes should not be run on the same day. 
The problem of assigning one of the six days of the work week to each route thus 
reduces to the problem of 6-coloring of the graph \cite{meyer}. 
In practice it might be 
desirable to have an approximately equal number of routes run on each of the six 
days. So we have to color the graph in an equitable way with six colors. Other 
applications of equitable coloring can be found in scheduling and timetabling. 

The notion of equitable colorability was introduced by Meyer \cite{meyer}. 
However, an earlier work of Hajnal and Szemer\'edi \cite{hfs:haj} showed that 
a graph $G$ with maximal degree $\Delta$ is equitably $k$-colorable 
if $k\geq\Delta+1$. Recently, Kierstead et al. \cite{fast} have given an $O(\Delta |V(G)|^2)$-time algorithm for equitable $(\Delta+1)$-coloring of graph $G$. In 1973, Meyer \cite{meyer} formulated the following 
conjecture:

\begin{conjecture} [Equitable Coloring Conjecture (ECC)]
For any connected graph $G$ other than a complete graph or an odd cycle, 
$\chi_{=}(G)\leq\Delta$.
\end{conjecture}

This conjecture has been verified for all graphs on six or fewer vertices. 
Lih and Wu \cite{kp} proved that the Equitable Coloring Conjecture is true for all bipartite graphs. Wang and Zhang \cite{hfs:wang} considered a broader class of graphs, namely $r$-partite graphs. They proved that Meyer's conjecture is true for complete graphs from this class. Also, the conjecture (or even the stronger one) was confirmed for outerplanar graphs \cite{hfs:yap} and planar graphs with maximum degree at least 13 \cite{planar}.

The \emph{corona} of two graphs $G$ and $H$ is a graph $G \circ H$ formed from one copy of $G$ and $|V(G)|$ copies of $H$ where the $i$th vertex of $G$ is adjacent to every vertex in the $i$th copy of $H$. For any integer $l \geq 2$, we define the graph $G \circ ^l H$ recursively from $G \circ H$ as $G \circ ^l H = (G \circ ^{l-1} H ) \circ H$ (cf. Fig.~\ref{rys1}). Graph $G \circ ^l H$ is also named as \emph{$l$-corona product} of $G$ and $H$. Such type of graph product was introduced by Frucht and Harary in 1970 \cite{frucht}. 

A straightforward reduction from graph coloring to equitable coloring by adding sufficiently many isolated vertices to a graph, proves that it is NP-complete to test whether a graph has an equitable coloring with a given number of colors (greater than two). Furmańczyk et al. \cite{hf} proved that the problem remains NP-complete for corona graphs. Bodlaender and Fomin \cite{treewidth} showed that equitable coloring problem can be solved to optimality in polynomial time for trees (previously known due to Chen and Lih \cite{cl}) and outerplanar graphs. A polynomial time algorithm is also known for equitable coloring of split graphs \cite{split}, cubic graphs \cite{cubic} and some coronas \cite{hf}. 

We will now briefly outline the remainder of the paper. In Section \ref{sek_com} we give an upper bound on the equitable chromatic number of $l$-corona product of graphs while in Section \ref{sek_rpartite} we give some results concerning the equitable colorability of $l$-corona products of some graphs and $r$-partite graphs. Next, in Section \ref{sek_cycle} we consider $l$-corona products of graphs $G$ with $\chi_=(G) \leq 4$ and cycles. In Section \ref{sek_path} we study $l$-corona products of these graphs $G$ and paths. In this way we extend the class of graphs that can be colored optimally in polynomial time and confirm the ECC conjecture \cite{hf}.

\begin{figure}[h!]
\begin{center}
\includegraphics[scale=0.9]{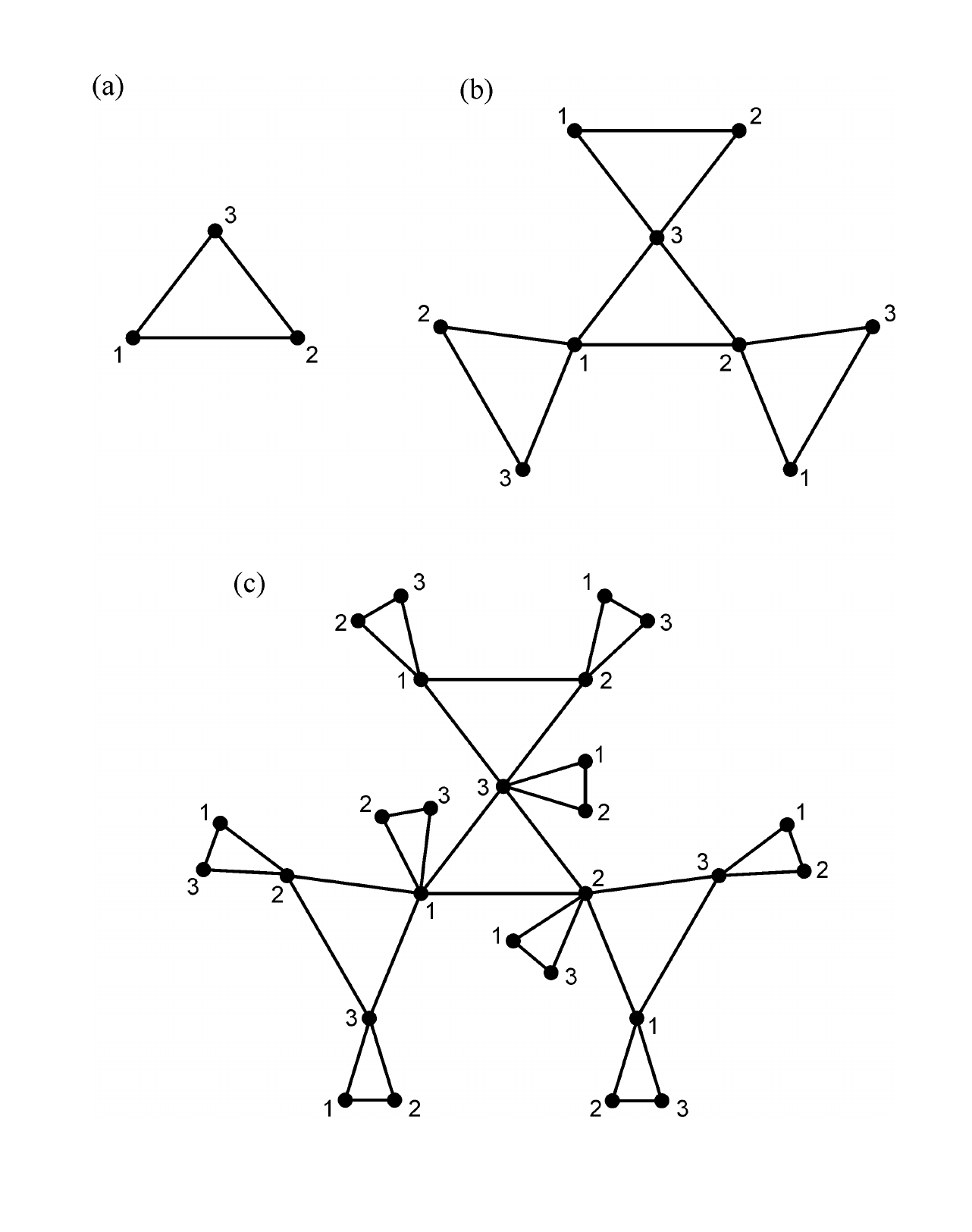}
\caption{Example of graphs: a) $C_3$; b) $C_3\circ K_2$; c) $C_3\circ^2 K_2$.}\label{rys1}
\end{center}
\end{figure}
\section{Equitable coloring of corona products with complete graphs}\label{sek_com}

It is known that $\chi_{=}(G \circ K_m)= m+1$ for every graph $G$ such that $\chi\left(G\right)\leq m+1$ \cite{hf}. Since the obtained graph $G \circ K_m$ is $(m+1)$-colorable, the graph $G \circ ^2 K_m$ is also equitably $(m+1)$-colorable, and so on. This result can be easy generalized to the $l$-corona product.

\begin{proposition}
If $G$ is a graph with $\chi\left(G\right)\leq m+1$, then  $\chi_{=}(G \circ ^l K_m)= m+1$, for any $l\geq 1$. \label{corcomplete}
\end{proposition}

Let us note that since $G$ is connected, the maximum degree of the corona $\Delta(G \circ ^ l K_m)$ is equal to $\Delta(G)+m \cdot l$. Since $m+1 \leq \Delta(G) + m \cdot l$, so ECC is true for such graphs.

Let us also notice that we immediately get an upper bound on the equitable chromatic number:
$$\chi_{=}(G \circ ^l H)\leq m+1,$$
where $l \geq 1$, $\chi(G)\leq m+1$ and graph $H$ is of order $m$.

\section{Equitable coloring of corona graphs with $r$-partite graphs} \label{sek_rpartite}
In this section we consider corona products of graph $G$ and $r$-partite graphs, where $G$ fulfills some additional conditions.

\begin{theorem}
Let $G$ be an equitably $k$-colorable graph on $n$  vertices and let $H$ be a $(k-1)$-partite graph. If $k|n$, then for any $l \geq 1$
$$\chi_=(G \circ^l H) \leq k.$$ \label{lrpart}
\end{theorem}
\begin{proof}
The proof is by induction on $l$. 
\begin{description}
\item[\rm{Step 1:}] For $l=1$ the theorem holds due to the following.

Suppose $V(G)=V_1 \cup V_2 \cup \cdots \cup V_k$, where $V_1, \ldots, V_k$ are independent sets each of size $n/k$. This means that they form an equitable $k$-coloring of $G$.
For each vertex $z \in V(G)$, let $H_z = (X_1^z, \ldots, X_{k-1}^z, E^z)$ be the copy of $(k-1)$-partite graph $H=(X_1,\ldots, X_{k-1}, E)$ in $G \circ H$ corresponding to $z$. Let
$$
\begin{array}{l}
V_1' = V_1 \cup \bigcup _{z\in V_2} X_1^z \cup \cdots \cup \bigcup _{z \in V_k} X_{k-1}^z,\\
V_2' = V_2 \cup \bigcup _{z\in V_3} X_1^z \cup \cdots \cup \bigcup _{z \in V_k} X_{k-2}^z \cup \bigcup _{z \in V_1} X_{k-1}^z ,\\
\vdots\\
V_{k-1}' = V_{k-1} \cup \bigcup _{z\in V_k} X_1^z\cup \bigcup _{z \in V_1} X_{2}^z \cup \cdots \cup \bigcup _{z \in V_{k-2}} X_{k-1}^z,\\
V_k' = V_k \cup \bigcup _{z\in V_1} X_1^z \cup \cdots \cup \bigcup _{z \in V_{k-1}} X_{k-1}^z.
\end{array}
$$

It is easy to see that $V(G\circ H)=V_1' \cup \cdots \cup V_k'$ is an equitable $k$-coloring of $G \circ H$. In this coloring each of $k$ colors is used exactly $n(1+|X_1|+\cdots + |X_{k-1}|)/k$ times. 

\item[\rm{Step 2:}] Suppose Theorem \ref{lrpart} holds for some $l \geq 1$.
\item[\rm{Step 3:}] We have to show that $\chi_=((G \circ^l H)\circ H) \leq k.$ Let us note that if $k|n$ then the cardinality of vertex set of $G \circ^l H$, which is equal to $n(|V(H)|+1)^l$, is also divisable by $k$. So using the inductive hypothesis  we get immediately the thesis.
\end{description}
\end{proof}

Let us note that the bound on the equitable chromatic number given in Theorem \ref{lrpart} holds for corona multiproducts of equitably $k$-colorable graph $G$ on $n$ vertices and $r$-partite graph $H$ where $r \leq k-1$ and $k|n$. Let assume that $G$ fulfills the assumption of Theorem \ref{lrpart} and graph $H$ is $r$-partite where $r < k-1$. We can add extra edges to graph $H$, without adding new vertices, until a new graph $H'$ is $(k-1)$-partite. Then, by Theorem \ref{lrpart} corona $G \circ ^l H'$ is equitably $k$-colorable. Since $G \circ ^l H$ is a subgraph of $ G \circ ^l H'$, it is also equitably $k$-colorable.
\begin{corollary}
Let $G$ be an equitably $k$-colorable graph on $n$  vertices and let $H$ be an $r$-partite graph where $r \leq k-1$. If $k|n$, then for any $l \geq 1$
$$\chi_=(G \circ^l H) \leq k.$$ \label{corrpart}
\end{corollary}

\section{Equitable coloring of corona products with cycles} \label{sek_cycle}

We will consider two cases: the first one for even cycles, and the second one for odd cycles.

\begin{theorem}
Let $G$ be an equitably 3-colorable graph on $n \geq 2$ vertices and let $k \geq 2$,  $l \geq 1$. If $3|n$ or $k=2$, then 
$$\chi_=(G \circ^l C_{2k}) = 3.$$ \label{k2l3dividen}
\end{theorem}

\begin{proof}
The first part of the theorem, for $3|n$, follows from Theorem \ref{lrpart}. Of course, we cannot use fewer than three colors.

The case when $k=2$ was partially considered in \cite{hf}.  
The authors proved that if $G$ is an equitably 3-colorable graph on $n \geq 2$ vertices, then $\chi_=(G \circ C_4) = 3$. This means that our theorem is true for $l=1$. The farther part of this proof is by induction on the number $l$, similarily as it was in the proof of Theorem \ref{lrpart}.
\end{proof}

We also know that in the remaining cases, i.e. when $G$ is equitably 4-colorable or $3 \nmid n$, we need more than three colors for equitable coloring of $G \circ C_{2k}$ \cite{hf}. 

\begin{theorem}
Let $G$ be an equitably 4-colorable graph on $n \geq 2$ vertices and let $k\geq 3$, $l \geq 1$. Then
$$\chi_=(G \circ^l C_{2k}) \leq 4.$$ \label{evencycle}
\end{theorem}
\begin{proof} Let us consider two cases.
\begin{description}
\item[\rm{Case 1:}] $3 | n$

We consider two subcases, depending on $n \bmod 4$.
\begin{description}
\item[\rm{Subcase 1.1:}] $n \bmod 4 = 0$.

The thesis follows immediately from Corollary \ref{corrpart}.
\item[\rm{Subcase 1.2:}] $n \bmod 4 \neq 0$.

First, we will show that our theorem is true for $l=1$ and then by induction on $l$ we will get the thesis for multicoronas $G \circ ^l C_{2k}$, $l \geq 1$.

Suppose graph $G$ has been colored equitably with 4 colors and the cardinalities of color classes are arranged in non increasing way. We order the vertices of $G$: $v_1, v_2, \ldots, v_n$ in such a way that vertex $v_i$ is colored with color $i \bmod 4$, and we use color 4 instead of 0. 

We color first $4x$ copies of $C_{2k}$ using in the $i$th copy $k$ times color $(i \bmod 4 +1) \bmod 4$, $\lceil k/2 \rceil$ times color $(i \bmod 4 +2) \bmod 4$ and $\lfloor k/2 \rfloor$ times color $(i \bmod 4 +3) \bmod 4$, where $x$ is defined below. In this part each color is used the same number of times. 

\begin{enumerate}[(i)]
\item $n \bmod 4 = 1$.

Since $n$ is a multiple of three, there is $p\geq 0$, such that $n=12p+9$. In this subcase $x$ is defined as $4p+1$. Finally, we have to color vertices in last five copies of $C_{2k}$ in corona $G \circ C_{2k}$ as follows:

\begin{itemize}
\item we color the 1st copy using $k$ times color 2,  $\lceil k/2 \rceil$ times color 3, $\lfloor k/2 \rfloor$ times color 4,
\item we color the 2nd copy using  $k$ times color 1,   $\lceil k/2 \rceil$ times color 4, $\lfloor k/2 \rfloor$ times color 3,
\item we color the 3rd copy using  $k$ times color 1,  $k$ times color 4,
\item we color the 4th copy using  $k$ times color 2,   $\lceil k/2 \rceil$ times color 3, $\lfloor k/2 \rfloor$ times color 1,
\item we color the 5th copy using $k$ times color 3,   $\lceil k/2 \rceil$ times color 4, $\lfloor k/2 \rfloor$ times color 2.
\end{itemize}

\item $n \bmod 4 = 2$.

Since $n$ is a multiple of three, there is $p\geq 0$, such that $n=12p+6$. In this subcase $x$ is also defined as $4p+1$. Finally, we have to color vertices in last two copies of $C_{2k}$ in corona $G \circ C_{2k}$ as follows:
\begin{itemize}
\item we color the 1st copy using $k$ times color 2,  $\lceil k/2 \rceil$ times color 3, $\lfloor k/2 \rfloor$ times color 4,
\item we color the 2nd copy using  $k$ times color 1,   $\lceil k/2 \rceil$ times color 4, $\lfloor k/2 \rfloor$ times color 3.
\end{itemize}

\item $n \bmod 4 = 3$.

Since $n$ is a multiple of three, there is $p\geq 0$, such that $n=12p+3$. In this subcase $x$ is defined as $4p$. Finally, we have to color the vertices in the last three copies of $C_{2k}$ in corona $G \circ C_{2k}$ as follows:
\begin{itemize}
\item we color the 1st copy using $k$ times color 3,  $\lceil k/2 \rceil$ times color 4, $\lfloor k/2 \rfloor$ times color 2,
\item we color the 2nd copy using  $k$ times color 1,   $\lceil k/2 \rceil$ times color 4, $\lfloor k/2 \rfloor$ times color 3,
\item  we color the 3rd copy using  $k$ times color 2, $\lceil k/2 \rceil$ times color 4, $\lfloor k/2 \rfloor$ times color 1.
\end{itemize}
\end{enumerate}
It can be easily checked that each of the above colorings, in all subcases, is an equitable 4-coloring of $G\circ C_{2k}$.
\end{description}

\item[\rm{Case 2:}] $3 \nmid n$.

It follows from \cite{hf} that if $G$ is an equitably 4-colorable graph on $n$ vertices, $n \geq 2$ and $3 \nmid n$, then $\chi_=(G \circ C_{2k}) = 4$ for $k \geq 3$. This means that our theorem is true for $l=1$. Therefore, by induction on $l$, we get the desired result.

\end{description}
\end{proof}

It turns out that in the case when the number of vertices of graph $G$ is not divisible by three, the weak inequality becomes equality.

\begin{theorem}
Let $G$ be an equitably 4-colorable graph on $n \geq 2$ vertices and let $k\geq 3$, $l \geq 1$. If $3 \nmid n$, then
$$\chi_=(G \circ^l C_{2k}) = 4.$$ \label{vahan}
\end{theorem}
\begin{proof} All we need is the proof that we cannot use fewer colors.
 
If $\chi(G) =4$ then of course $\chi_=(G \circ ^l C_{2k})> 3$, for any $l$. Let us assume that $\chi(G) \leq 3$. Note that any 3-coloring of $G$ uniquely determines a 3-coloring of $G \circ ^l C_{2k}$, $l\geq 2$. When we color the vertices in a copy of $C_{2k}$ linked to a vertex of $G \circ ^{l-1} C_{2k}$, we use two available colors. It is not hard to notice that the difference between cardinalities of color clases is the smallest when 3-coloring of $G$ is strongly equitable. In our case, when $n$ is not divisible by three, a strongly equitable coloring does not exist. 
If the difference between cardinalities of every two color classes of $G$ is greater than or equal to one, any 3-coloring of $G \circ ^ l C_{2k}$ cannot be equitable. This follows from the following reasoning.

We claim that every equitable (not strongly) 3-coloring of $G$ determines 3-coloring of $G \circ^l C_{2k}$ with maximum difference among the color classes equaling to $(k-1)^l$, $l\geq 2$.

We prove the claim by induction on $l$. 
\begin{description}
\item[\rm{Step 1:}] $l=1$.
\begin{enumerate}[(i)]
\item  $n \bmod 3 = 1$

Cardinalities of color classes for colors 1, 2 and 3 are equal to $\lfloor n/3 \rfloor (2k+1)+1$,$\lfloor n/3 \rfloor(2k+1)+k$ and $\lfloor n/3 \rfloor(2k+1)+k$, respectively. The maximum difference between color classes is equal to $(k-1)^1$. Our claim holds.

\item $n \bmod 3 = 2$

Cardinalities of color classes for colors 1, 2 and 3 are equal to $\lfloor n/3 \rfloor(2k+1)+1+k$,$\lfloor n/3 \rfloor(2k+1)+1+k$ and $\lfloor n/3 \rfloor(2k+1)+2k$, respectively. The maximum difference between color classes is also equal to $(k-1)^1$. So, our claim holds also in this case.
\end{enumerate}

\item[\rm{Step 2:}] \emph{Induction hypothesis for $l\geq 1$.}

We assume that maximum difference between color classes in 3-coloring of $G \circ ^l C_{2k}$ is equal to $(k-1)^l$.
\item[\rm{Step 3:}] \emph{The proof that the difference for $l+1$ does not exceed $(k-1)^{l+1}$. }

It is easy to see that we have to compute the difference between cardinalities of color class for color  3 $(|C_3^{l+1}|)$ and color 1 ($| C_1^{l+1}|$). Again, we have to consider two subcases:
\begin{enumerate}[(i)]
\item $n \bmod 3 =1$

\begin{enumerate}[(a)]
\item $|C_1^l| = x$ for some $x$, and $|C_2^l|=|C_3^l| = x + (k-1)^l$.

Let us notice that $|C_1^{l+1}| = |C_1^l| + (|C_2^l|+|C_3^l|) \cdot k = x + 2xk + 2(k-1)^lk$ - we can color only these vertices of copies of $C_{2k}$ in $G \circ ^{l+1} C_{2k}$ with color one that are not adjacent to vertex colored with one in $G \circ ^ l C_{2k}$, while $|C_3^{l+1}| = x+(k-1)^l + (2x+(k-1)^l)k$. The difference $|C_3^{l+1}| - |C_1^{l+1}| = (k-1)^{l+1}$.

\item $|C_1^l| = x$ for some $x$, and $|C_2^l| = |C_3^l| = x - (k-1)^l$.

Analogously to above.
\end{enumerate}

\item $n \bmod 3 = 2$
\begin{enumerate}
\item $|C_1^l| =|C_2^l|= x$ for some $x$, and $|C_3^l| = x + (k-1)^l$.

Let us notice that $|C_1^{l+1}| = |C_1^l| + (|C_2^l|+|C_3^l|) \cdot k = x + 2xk + (k-1)^lk$ - we can color only these vertices of copies of $C_{2k}$ in $G \circ ^{l+1} C_{2k}$ with color one that are not adjacent to vertex colored with one in $G \circ ^ l C_{2k}$, while $|C_3^{l+1}| = x+(k-1)^l + 2xk$. The difference $|C_1^{l+1}| - |C_3^{l+1}| = (k-1)^{l+1}$.

\item $|C_1^l| =|C_2^l|= x$ for some $x$, and $|C3^l| = x - (k-1)^l$.

Analogously to above.
\end{enumerate}
\end{enumerate}
Summing up, even when $3|(2k+1)$ and $G$ is equitably (not strongly) 3-colorable, corona $G \circ^l C_{2k}$, $l \geq 1$ demands four colors to be equitably colored. Hence, our thesis holds.
\end{description}

\end{proof}

Now, we will consider cycles on odd number of vertices.

\begin{theorem}
 Let $G$ be an equitably 4-colorable graph on $n\geq 2$ vertices and let $k\geq 1$. Then for any $l \geq 1$ we have
$$\chi_=\left(G\circ^l C_{2k+1}\right)= 4.$$ \label{colC2k+1}
\end{theorem}
\begin{proof}
The inequality $$\chi_=\left(G\circ^l C_{2k+1}\right) \leq 4$$ is true due to Corollary \ref{corrpart}. Since $K_1 \circ C_{2k+1}$ is a subgraph of $G\circ^l C_{2k+1}$, we cannot use fewer colors.
\end{proof}

We have considered equitable coloring of corona product of graphs on at least two vertices and cycles. Now, for the sake of completeness, we consider  equitable colorings of corona products of one isolated vertex and cycles.
We have noticed in \cite{hf} that

\begin{equation}
\chi_=(K_1 \circ C_{m})= \left \{
\begin{array}{ll}
4, & \text{ if } m=3,\\
\left \lceil \frac{m}{2} \right \rceil +1, & \text{ if } m > 3.
\end{array}\right.\label{wzor}\end{equation}

The value of equitable chromatic number of multicorona $K_1 \circ ^l C_m$ can be arbitrarily large for $l=1$ (cf. Equality (\ref{wzor})). The situation changes significantly for larger values of $l$.

\begin{theorem}
Let $m \geq 3$ and $l \geq 2$. Then
$$ \chi_=(K_1 \circ ^l C_m) =
\left \{
\begin{array}{ll}
3, & \text{ if } m =4,\\
4, & \text{ otherwise.}
\end{array}\right.
$$ \label{c5}
\end{theorem}

\begin{proof}
Let us consider three cases.
\begin{description}
\item[\rm{Case 1:}] $m=3$ 

Since $C_3=K_3$, our thesis follows immediately from Proposition \ref{corcomplete}.

\item[\rm{Case 2:}] $m=4$

By Equality (\ref{wzor}) $\chi_=(K_1 \circ^1 C_{4})=3$. For $l\geq 2$ the truth of our theorem follows from Theorem \ref{k2l3dividen}. 

\item[\rm{Case 3:}] $m\geq 5$

Our proof is by induction on $l$.
\begin{description}
\item[\rm{Step 1:}] For $l=2$ the theorem holds due to the following. We consider two cases depending on the parity of $m$.
\begin{enumerate}[(i)]
\item $m$ is even.

First, we prove that $\chi_=(K_1 \circ ^2 C_{2k}) > 3$. Let us notice that 3-coloring of $K_1 \circ ^2 C_{2k}$ is unique up to permutation of colors. The vertex of $K_1$ is assigned, say, color 1, vertices of $C_{2k}$ in $K_1 \circ C_{2k}$ adjacent to vertex colored 1, must be colored with 2 and 3, alternately. Next, we assign two available colors to vertices in $2k+1$ copies of $C_{2k}$ in $K_1 \circ ^2 C_{2k}$. Cardinalities of color classes in such a 3-coloring are equal to $1+2k \cdot k$ and twice $k+ (k+1) \cdot k$, respectively. It is easy to see that this coloring is not equitable for $k \geq 3$.

Now, we have to prove that there is an equitable 4-coloring of $K_1 \circ ^2 C_{2k}$. The cardinalities of color classes in every such coloring are equal to $\lceil |V(K_1 \circ ^2 C_{2k})| / 4 \rceil = \lceil (2k+1)(2k+1)/4 \rceil$, $\lceil [(2k+1)(2k+1)-1]/4 \rceil$, $\lceil [(2k+1)(2k+1)-2]/4 \rceil$ and $\lceil [(2k+1)(2k+1)-3]/4 \rceil$, respectively.

First, we color the vertex of $K_1$ with color 1, next the vertices of $C_{2k}$ in $K_1 \circ C_{2k}$ with colors 2, 3 and 4 using them $\lceil 2k/3 \rceil$, $\lceil (2k-1)/3 \rceil$ and $\lceil (2k-2)/3 \rceil$ times, respectively. Finally, we color vertices in $2k+1$ copies of $C_{2k}$ using each time three allowed colors. In each copy every allowed color is used  $\lceil 2k/3 \rceil$ or  $\lfloor 2k/3 \rfloor$ times. One can verify that such equitable 4-coloring of $K_1 \circ ^2 C_{2k}$ exists for each $k\geq 3$.

\item $m$ is odd.

Let us notice that $|V(K_1 \circ ^2 C_{2k+1})|=(2k+2)(2k+2)=4(k+1)^2$. This means that each of four colors must be used exactly $(k+1)^2$ times in every equitable coloring.

Graph $K_1 \circ ^2 C_{2k+1}$ consists of $2k+2$ copies of $C_{2k+1}$ joined to vertices of $K_1 \circ C_{2k+1}$ appropriately. The equitable 4-coloring of $K_1 \circ ^2 C_{2k+1}$ is formed as follows:
\begin{itemize}
\item the vertex of $K_1$ is colored with 1
\item the remaining vertices of $K_1 \circ C_{2k+1}$ are assigned colors 2, 3 and 4 with cardinalities equal to $k$, $k$ and 1, respectively
\item the copy of $C_{2k+1}$ in $K_1 \circ ^2 C_{2k+1}$ joined to vertex colored 1 is assigned colors 2, 3 and 4 with cardinalities equal to 1, $k$ and $k$, respectively
\item copies of  $C_{2k+1}$ in $K_1 \circ ^2 C_{2k+1}$ joined to vertex colored 2 are assigned colors 1, 3 and 4 with cardinalities in each cycle equal to 1, $k$ and $k$, respectively
\item copies of  $C_{2k+1}$ in $K_1 \circ ^2 C_{2k+1}$ joined to vertex colored 3 are assigned colors 1, 2 and 4 with cardinalities in each cycle equal to $k$, $k$ and 1, respectively
\item copies of  $C_{2k+1}$ in $K_1 \circ ^2 C_{2k+1}$ joined to vertex colored 4 are assigned colors 1, 2 and 3 with cardinalities in each cycle $C_{2k+1}$ equal to $k$, $k$ and 1, respectively
\end{itemize}

In such a coloring every color is used exactly $(k+1)^2$ times. Let us notice that we cannot use fewer colors, because each graph $K_1 \circ ^2 C_{2k+1}$ includes odd cycle (3-colorable) in which every vertex is joined to the vertex of $K_1$. Hence $\chi_=(K_1 \circ ^2 C_{2k+1}) = 4$.
\end{enumerate}

\item[\rm{Step 2:}] \emph{Induction hypothesis.} Suppose Theorem \ref{c5} holds for some $l\geq 2$.

\item[\rm{Step 3:}] We have to show that $\chi_=((K_1 \circ ^l C_m) \circ C_m)=4$. 

We consider two cases.
\begin{enumerate}[(i)]
\item $m$ is even.

The inequality
$$\chi_=(K_1 \circ ^l C_m) \leq 4$$ is true due to Theorem \ref{evencycle}. We have to proof that the equitable coloring with smaller number of colors does not exist. 

Let us assume, to the contrary, that $\chi_=(K_1 \circ ^l C_{2k}) \leq 3$, $k \geq 3$, $l \geq 2$. Since $\chi(K_1 \circ ^l C_{2k}) = 3$ (graph $K_1 \circ ^l C_{2k}$ contains $K_3$), then we only need to fix equitable 3-coloring. Let us notice that every 3-coloring of $K_1 \circ ^l C_{2k}$ is unique up to permutation of colors. We again claim that in such a 3-coloring the difference between color classes is equal to $(k-1)^l$. The proof of our claim is analogous to the proof of Theorem \ref{vahan}.

It follows that such a coloring is not equitable for $k\geq 3$, a contradiction. 
\item $m$ is odd. 

The thesis follows from Theorem \ref{colC2k+1}.
\end{enumerate}
\end{description}
\end{description}
\end{proof}

\section{Equitable coloring of corona products with paths}\label{sek_path}

Since graph $G \circ ^ l P_m$, $m\geq 2$ contains a triangle and it is a subgraph of $G \circ ^ l C_m$, and $P_m$ is bipartite, Theorems \ref{lrpart} and \ref{k2l3dividen} imply

\begin{corollary}
Let $G$ be an equitably 3-colorable graph on $n$ vertices, and let $m, n, l \geq 2$. If $m=4$ or $3|n$, then
$$\chi_=(G \circ ^l P_{m}) = 3.$$\label{path4}
\end{corollary}

It turns out that there are more graphs (corona multiproduct with paths) that can be equitably colored with three colors.

\begin{theorem}
Let $G$ be an equitably 3-colorable graph on $n \geq 2$ vertices and let $l \geq 1$. If $m=2, 3, 5$, then
$$\chi_=(G \circ^l P_{m}) = 3.$$\label{path25}
\end{theorem}
\begin{proof}
The authors proved \cite{hf} that if $G$ is an equitably 3-colorable graph on $n \geq 2$ vertices, then $\chi_=(G \circ P_{m}) = 3$ for $m=2, 3, 5$. This means that our theorem holds for $l=1$. The farther part of the proof is by induction on the number $l$.
\end{proof}

In the remaining cases of $m$ we sometimes have to use four colors. 

Since $G \circ^l P_m$, is a subgraph of $G \circ ^l C_m$, using Theorem \ref{evencycle} we get the following
\begin{corollary}
 Let $G$ be an equitably 4-colorable graph on $n\geq 2$ vertices and let $l \geq 1$, $m\geq 6$. Then 
$$\chi_ =\left(G\circ ^l P_{m}\right)\leq 4.$$ \label{wnsciezki}
\end{corollary}

Now, we will consider equitable coloring of corona of $K_1$ and paths.
Since $K_1 \circ P_{m}$ is a fan, we have:
$$\chi_=(K_1 \circ P_{m})=\left\lceil \displaystyle\frac{m}{2}\right\rceil+1$$
for $m \geq 3$.

Since paths are subgraphs of cycles, the following corollary holds.

\begin{corollary}
Let $m,l\geq 2$. Then
$$\chi_=(K_1 \circ ^ l P_m) 
\left \{
\begin{array}{ll}
=3, \text{ if } m=4\\
\leq 4, \text{ otherwise.}
\end{array}\right.$$\label{corpaths}
\end{corollary}

We can precise this result as follows.

\begin{theorem}
Let $m,l \geq 2$. Then
$$\chi_=(K_1 \circ ^ l P_m) 
\left \{
\begin{array}{ll}
= 3, \text{ if } m=2,3,4,5\\
= 4, \text{ if } m \geq 6 \text{ and } $m$ \text{ even}\\
\leq 4, \text{ otherwise.}
\end{array}\right.$$
\end{theorem}
\begin{proof}
\begin{description}
\item[\rm{Case 1:}] $m=2,3,4,5$

Let us start from $l=2$.
\begin{enumerate}[(i)]
\item $m=2$

Any 3-coloring of $K_1 \circ ^ 2 P_2$ leads to cardinalities of all color classes equal to 3. 
\item $m=3$

We are able to color $K_1 \circ ^2 P_3$ with three colors in such a way that cardinalities of color classes are equal to 6, 5 and 5, respectively. 

\item $m=4$

This case was considered in Corollary \ref{corpaths}.

\item $m=5$

We are able to color $K_1 \circ ^2 P_5$ with three colors in such a way that cardinalities of all color classes are equal to 12.
\end{enumerate}

Using Corollary \ref{path4} and Theorem \ref{path25} we get our thesis for $l \geq 2$.

\item[\rm{Case 2:}] \emph{$m \geq 6$ and $m$ is even.}

According to Corollary \ref{wnsciezki} we only have to prove that we cannot use a smaller number of colors. The argument is the same as in the proof of Theorem \ref{c5}.
\end{description}

\end{proof}

\section{Conclusion} \label{sek_conc}

In the paper we have given some results concerning $l$-corona products that confirm  Equitable Coloring Conjecture. It turns out that the ECC conjecture follows for every $l$-corona product $G \circ ^l H$, where graph $H$ is on $m$ vertices and graph $G$ can be properly colored with $m-1$ colors. Moreover we have established some special cases of $l$-corona products $G \circ ^l H$ that can be colored with 3 or 4 colors efficiently.  
The main of our results are summarized in Table \ref{pods}.
\begin{table}[htb]
\begin{center}
\begin{tabular}{|c|*{7}{c|}}\hline
\multicolumn{2}{|c|}{\multirow{2}{*}{\backslashbox[60mm]{$G$}{$H$}}} & 

\multicolumn{1}{|c|}{\makebox[3em]{bipartite}} & \multicolumn{2}{|c|}{\makebox[5em]{even cycles $C_{2k}$}}&
\multicolumn{1}{|c|}{\makebox[0.2em]{odd}}
& \multicolumn{2}{|c|}{\makebox[5em]{paths $P_k$}}
\\\cline{4-5} \cline{7-8}
\multicolumn{1}{|c}{}& & graphs
&$k=2$ & $k \geq 3$ &cycles& $2 \leq k \leq 5$  &  $ k \geq 6$ \\
\hline\hline

equitably 3-colorable  & $3|n$ &3&\multirow{2}{*}{3} & 3 & \multirow{2}{*}{4}&\multirow{2}{*}{3}&3\\ 
\cline{2-3} \cline{5-5} \cline{8-8}
graph $G$ on $n \geq 2$ vertices & $3 \nmid n$ &&&4&&&4\\
\hline
equitably 4-colorable 
 & $3|n$& $\leq 4$ &\multirow{2}{*}{$\leq 4$}&$ \leq 4$ &\multirow{2}{*}{4} & \multirow{2}{*}{$ \leq 4$} & \multirow{2}{*}{$ \leq 4$}\\ 
\cline{2-3} \cline{5-5} 
graph $G$ on $n \geq 2$ vertices &$3\nmid n$ & & &4 &&&\\
\hline
\end{tabular}
\caption{Possible values of the equitable chromatic number of coronas $G \circ ^ l H$. }\label{pods}
\end{center}
\end{table}

Of course, the complexity of equitable coloring of $G \circ^l H$ depends on the complexity of equitable 3- or 4-coloring of graph $G$, which is generally NP-hard. More precisely, since the time spend to color any vertex of $H$ is constant, such a coloring of graphs under consideration can be done in time $O(g(n) \cdot n')$, where $g(n)$ is the complexity of equitable 3- or 4-coloring of $n$-vertex graph $G$ and $n'$ is the number of vertices in the remaining part of $G \circ ^l H$. However, the following graphs: 
\begin{itemize}
\item broken spoke wheels \cite{m}, 
\item reels \cite{m},
\item cubic graphs except $K_4$ \cite{cubic},
\item some graph products \cite{furm, prod}
\end{itemize}
admit equitable 3-coloring in polynomial time, and so do the corresponding multicoronas.


\begin{thebibliography}{8}

\bibitem{treewidth} H.L. Bodleander, F.V. Fomin, Equitable colorings of bounded treewidth graphs, \emph{Theor. Comput. Sci.} 349(1) (2005), 22--30. 

\bibitem{split} B.L. Chen, M.T. Ko and K.W. Lih, Equitable and $m$-bounded coloring of split graphs, in: \emph{Combinatorics and Computer Science} (Brest, 1995)  LCNS 1120, Springer (1996), 1--5.

\bibitem{cl} B.L. Chen and K.W. Lih, Equitable coloring of trees, \emph {J. Combin. Theory Ser. B} 61 (1994), 83--87.
  
\bibitem{cubic} B.L. Chen, K.W. Lih and P.L. Wu, Equitable coloring and the maximum degree, \emph {Europ. J. Combinatorics} 15 (1994), 443--447.

\bibitem{frucht} R. Frucht, F. Harary, On the corona of two graphs,  \emph{Aequationes Math.} 4 (1970), 322--325.

\bibitem{m} H. Furmańczyk, Equitable coloring of graphs, in: \emph{Graph Colorings}, (M. Kubale, ed.) American Mathematical Society Providence, Rhode Island (2004). 

\bibitem{furm} H. Furmańczyk, Equitable coloring of graph products, \emph{Opuscula Mathematica} 26(1) (2006), 31--44.

\bibitem{hf} H. Furmańczyk, K. Kaliraj, M. Kubale, V.J. Vivin, \emph{Equitable coloring of corona products of graphs}, submitted.

\bibitem{hfs:haj} A. Hajnal, E. Szemeredi, Proof of a conjecture of Erdös, 
in: \emph{Combinatorial Theory and Its Applications, II} 601--623,
Colloq.\ Math.\ Soc.\ János Bolyai, Vol. 4, North-Holland,
Amsterdam, (1970).

\bibitem{fast} H.A. Kierstead, A.V. Kostochka, M. Mydlarz, E. Szemeredi, A fast algorithm for equitable coloring, \emph{Combinatorica} 30(2) (2010), 217--224.
 
\bibitem{kp}  K.W. Lih, P.L. Wu, On equitable coloring of bipartite graphs, \emph {Disc. Math.} 151 (1996), 155--160.

\bibitem{prod} W.H. Lin, G.J. Chang, Equitable colorings of Cartesian products of graphs, \emph{Disc. App. Math.} 160 (2012), 239--247.

\bibitem{meyer} W. Meyer, Equitable coloring, \emph{Amer. Math. Monthly} 80 (1973), 920--922.

\bibitem{hfs:wang} W. Wang, K. Zhang, Equitable colorings of line graphs 
and complete $r$-partite graphs, \emph{Systems Science and Mathematical 
Sciences} 13 (2000), 190--194.

\bibitem{hfs:yap} H.P. Yap, Y. Zhang, The Equitable $\Delta$-Coloring Conjecture holds for outerplanar graphs, \emph{Bulletin of the Inst.\ of 
Math.\ Academia Sinica} 25 (1997), 143--149.

\bibitem{planar} H.P. Yap, Y. Zhang, Equitable colourings of planar graphs, \emph{J. Comb. Math. Comb. Comput.} 27 (1998), 97--105.
\end{thebibliography}
\end{document}